\DocumentMetadata{
  lang          = en,
  pdfstandard   = {ua-2,a-4f},
  tagging       = on,
  tagging-setup = {math/setup=mathml-SE} 
}
\documentclass[10pt,twocolumn,a4paper]{article}
\usepackage[margin=1.5cm,columnsep=1cm]{geometry}
\usepackage[T1]{fontenc}

\usepackage{amsmath,amsthm}
\usepackage{cite,microtype,hyperref}
\usepackage[cochineal]{newtx}
\usepackage{titlesec}
\titleformat{\section}{\large}{\thesection}{1em}{}
\newtheorem{theorem}{Theorem}
\newtheorem{corollary}{Corollary}

\title{Sudoku Grids That Require Many Clues}

\author{David Eppstein\footnote{Research supported in part by NSF grant CCF-2212129.} \qquad\qquad Xinyu (Cindy) Zhang\\
\normalsize Computer Science Department, University of California, Irvine}
\date{ }

\begin{document}

\maketitle
\thispagestyle{empty}

\begin{abstract}
Motivated by worst-case algorithmic time bounds for solving sudoku, we prove that a majority of filled-in $n^2\times n^2$ sudoku grids require all but a logarithmic fraction of cells to be filled by clues. For $9\times 9$ and $16\times 16$ sudoku, we construct grids that require $18$ clues and $80$ clues.
\end{abstract}

\section{Introduction}

In sudoku puzzles, a $9\times 9$ grid of square cells, partitioned into nine $3\times 3$ blocks,
is to be filled with the digits from $1$ to $9$ so each row, column, and block contains each digit. To make the solution unique,
some cells are initially filled as clues. Often $30\%$--$35\%$ of cells are clues, but many fewer can be needed: the minimum number of clues in a uniquely-solvable puzzle is $17$. Sudoku can be generalized to $n^2\times n^2$ grids partitioned into $n^2$ blocks of size $n\times n$, and $16\times 16$ sudoku is conjectured to require $56$ clues~\cite{mcguire:2013}. For $n^2\times n^2$ sudoku the minimum number of clues is conjecturally $\Omega(n^4)$; only $\Omega(n^2)$ has been proven, but special clue sets that allow a solution using very restricted deduction rules require $\Omega(n^4)$ clues. Further, some sudoku puzzles can be uniquely specified using $\lfloor n^4/4\rfloor$ clues~\cite{mahmah16}.

We study a variant of this minimum-clue problem. Instead of the fewest clues among all sudoku puzzles, we seek the fewest clues for a given filled-in grid, and a grid that maximizes this minimum number of clues.  We prove that most $n^2\times n^2$ filled grids require $n^4-O(n^4/\log n)$ clues, leaving only a logarithmic fraction of cells unfilled. We construct many $9\times 9$ sudoku grids requiring $18$ clues, matching a report of McGuire et al.~\cite{mcguire:2013} that $17$ clues do not always suffice, and a $16\times 16$ grid requiring $80$ clues, well above the conjectured $56$-clue minimum. Our results are motivated by the worst-case algorithmic complexity of sudoku, and by a result that an $n^2\times n^2$ sudoku puzzle with $m$ clues can be solved in time $O(n^4\,2^{n^4-m})$. For this solution method, grids with more clues can be solved more quickly.

A related problem, adding the fewest clues to a partial set of clues to produce a unique solution, is hard for the second level of the polynomial hierarchy~\cite{DemMaSch-TCS-18}. Our problem differs in seeking the fewest clues for the solution itself. It would be of interest to find its computational complexity; we do not.

\section{Algorithmic motivation}

Published sudoku puzzles often can be solved by deduction rules of polynomial complexity~\cite{eppstein:2005}, but sudoku is $\mathsf{NP}$-hard and hence unlikely to have polynomial worst-case time~\cite{yato:2003}. Backtracking often works well but this has not been validated theoretically. To close the gap between theory and practice, we seek worst-case times that are exponential but as fast as possible. A brute force search that checks all ways of filling empty cells with digits would take time $O(n^4 n^{2(n^2-m)})$. Instead, we adapt an algorithm by Björklund, Husfeldt, and Koivisto that uses dynamic programming and inclusion--exclusion on the lattice of subsets of a given set to find exact covers~\cite{bjorklund:2009}.

Define a \emph{valid placement} of a digit to be a set of cells that, together with the clues for that digit, include one cell in each row, column, and block of the puzzle. A sudoku solution partitions all cells into valid placements, one per digit. This is not quite an exact cover because each digit must contribute exactly one valid placement to the cover. Instead we use the \emph{sum weighted partitions} problem defined by Björklund et al.~\cite{bjorklund:2009}. It takes as input a family of $k$ functions $f_i$ from subsets of a given set to numbers; the output is the sum of products $\sum_P\prod_i f_i(P_i)$ where the sum is over partitions $P$ into $k$ sets $P_i$. For sudoku, let $k=n^2$ and let $f_i(S)$ be one when $S$ is a valid placement of digit $i$ and zero otherwise; then the sum weighted partitions value is the number of puzzle solutions. (Conventionally this should equal one but we do not assume this.) By evaluating this number of solutions, we can determine whether any digit placement leads to a solution, and by iterating, solve the puzzle.

\begin{theorem}
We can solve sudoku in time $n^{O(1)}\,2^{n^4-m}$.
\end{theorem}

\begin{proof}
We construct the functions $f_i$ defined above by setting all subset values to zero, listing valid placements using a reduction to enumerating paths in graphs~\cite{eppstein:2012} and setting each placement's value to one. We then apply the  sum weighted partitions algorithm of Björklund et al., which for $n^2$ 0--1 valued functions on subsets of $n^4-m$ cells takes the stated time.
\end{proof}

\section{Existence}

We prove the existence of filled grids with many clues by comparing the number of filled grids to the number of ways of specifying clues. If $S(n)$ denotes the number of filled $n^2\times n^2$ sudoku grids, then only $S(2)=288$ and \[S(3)=6\,670\,903\,752\,021\,072\,936\,960\] are known exactly~\cite{oeis}. Asymptotically, Keevash~\cite{keevash:2019} proves that \[S(n)=\left(\frac{n^2}{e^3}+o(n^2)\right)^{n^4}.\]
If $C(n,m)$ denotes the number of ways of specifying $m$ clues in an $n^2\times n^2$ sudoku grid, then
\[C(n,m)\le\binom{n^4}{m}n^{2m},\]
for there are $\tbinom{n^4}{m}$ choices of clue placements and $(n^2)^m=n^{2m}$ ways of filling them with digits.
At most $C(n,m)$ filled grids can be determined by these $m$ clues.
It does not increase the count of grids to include puzzles with fewer clues, because such a puzzle could be extended to exactly $m$ clues by adding unnecessary clues.
Comparing $S(n)$ and $C(n,m)$ gives our result:

\begin{theorem}
All but a $1/2^{n^4}$ fraction of the filled $n^2\times n^2$ sudoku grids require $m\ge n^4-O(n^4/\log n)$ clues.
\end{theorem}

\begin{proof}
We need $C(n,m)\ge S(n)/2^{n^4}$ to represent more than a $1/2^{n^4}$ fraction of the filled grids.
Taking logarithms (base $2$) of both sides, and using the fact that $\tbinom{n^4}{m}\le 2^{n^4}$, gives
\[n^4+2m\log_2 n \ge 2n^4\log_2 n-(3\log_2 e+1)n^4+o(n^4).\]
Subtracting $n^4$ from both sides, dividing by $2\log_2 n$, and consolidating the lower-order terms gives the stated bound.
\end{proof}

As each filled grid has at least one minimal clue set, the same counting argument shows that the same fraction of minimal clue sets require the same number of clues.
This leads to a better exponent in the exponential-time algorithms for these puzzles:

\begin{corollary}
For all but a $1/2^{n^4}$ fraction of filled $n^2\times n^2$ sudoku grids or minimal sudoku clue sets, solving a puzzle with this solution or clue set takes time $2^{O(n^4/\log n)}$. The same time bound holds in the average case over random filled grids or minimal clue sets.
\end{corollary}

\section{Construction}

A Latin square is like a sudoku grid in having distinct digits in each row or column of its grid of cells, but its side length need not be square and its cells are not grouped into blocks. The following steps pack $n^2$ Latin squares of size $n\times n$ into an $n^2\times n^2$ sudoku grid:
\begin{itemize}\itemsep0pt
\item Partition the $n^2$ digits into $n$ subsets $D_i$ of $n$ digits. In the filled sudoku grid, each row of each block will use one of these subsets.
\item For each of $n$ horizontal groups of $n\times n$ sudoku blocks, use an $n\times n$ Latin square to determine which row of each block uses which digit subset $D_i$.
\item For each of $n$ vertical groups of $n\times n$ sudoku blocks, and each subset $D_i$ of digits, place $D_i$ into its assigned rows (an $n\times n$ subset of cells) as an $n\times n$ Latin square.
\end{itemize}
Let $L(n)$ denote the number of $n\times n$ Latin squares.
There are $n^2!/n!^{n+1}$ partitions of digits into subsets $D_i$ (treating any permutation of subsets as equivalent), $L(n)^n$ ways of choosing Latin squares for horizontal groups of blocks, and $L(n)^{n^2}$ ways of choosing Latin squares for digit subsets in each vertical group of blocks. Thus, the total number of choices is \[\frac{n^2!}{n!^{n+1}}L(n)^{n^2+n}.\]

$L(3)=12$, so this formula gives $2\,496\,508\,125\,511\,680$
distinct $9\times 9$ sudoku grids, each packed with nine $3\times 3$ Latin squares. Each Latin square requires two clues, because a single clue would allow swapping the other two digits in that square. Thus, each of these sudoku grids requires at least $18$ clues.

For $n=4$, case analysis shows that the Latin square
\[
\begin{matrix}
1 & 2 & 3 & 4 \\
2 & 1 & 4 & 3 \\
3 & 4 & 1 & 2 \\
4 & 3 & 2 & 1 \\
\end{matrix}
\]
requires five clues. Using our construction to pack $16$ copies of this Latin square into a $16\times 16$ sudoku grid produces a filled grid requiring $80$ clues.

Hatami and Qian~\cite{Hatami:2018} prove that all $n\times n$ Latin squares require $\Omega(n^2)$ clues, so in general the sudoku grids from our construction require $\Omega(n^4)$ clues.

\linespread{1}\normalfont
\bibliographystyle{abbrvurl}
\bibliography{references}

\end{document}